\documentclass[10pt,twocolumn]{IEEEtran}

\usepackage{rotate,epsf,epsfig,url,subfigure,wrapfig,algorithm2e}
\usepackage[sc]{mathpazo}
\usepackage{amsfonts,amssymb,graphicx}
\usepackage{psfrag}
\usepackage{amsthm}
\usepackage[usenames,dvipsnames]{color}
\usepackage{palatino}
\usepackage{epsfig}
\usepackage{amsmath,nicefrac}
\usepackage{latexsym}
\usepackage{amsfonts}
\usepackage{amssymb}
\usepackage{calrsfs}
\usepackage{enumitem}

\usepackage{url}
\usepackage{subfigure}
\usepackage{wrapfig}
\usepackage{wasysym}   % has double closed integral sign

\usepackage{soul}
\usepackage[numbers,sort&compress]{natbib}

\makeatother
\newcommand{\va}{\mbox{${\bf a}$}}

\newcommand{\vx}{\mbox{${\bf x}$}}

\newcommand{\vh}{\mbox{${\bf h}$}}

\newcommand{\ve}{\mbox{${\bf e}$}}

\newcommand{\vw}{\mbox{${\bf w}$}}

\newcommand{\vwt}{\mbox{${\bf \tilde w}$}}

\newcommand{\vy}{\mbox{${\bf y}$}}
\newcommand{\vz}{\mbox{${\bf z}$}}

\newcommand{\vone}{\mbox{${\bf 1}$}}

\newcommand{\mA}{\hbox{{\bf A}}}

\newcommand{\mD}{\hbox{{\bf D}}}

\newcommand{\mG}{\hbox{{\bf G}}}

\newcommand{\mH}{\hbox{{\bf H}}}

\newcommand{\mHt}{\mbox{${\bf \tilde H}$}}
\newcommand{\mI}{\hbox{{\bf I}}}

\newcommand{\mQ}{\hbox{{\bf Q}}}
\newcommand{\mR}{\mbox{{$\bf R$}}}

\newcommand{\mS}{\mbox{{$\bf S$}}}

\newcommand{\mW}{\hbox{{\bf W}}}

\newcommand{\grg}{\gamma}
\newcommand{\gd}{\delta}
\newcommand{\gre}{\varepsilon}

\newcommand{\gth}{\theta}

\newcommand{\gl}{\lambda}

\newcommand{\gs}{\sigma}

\newcommand{\gt}{\tau}

\def\bm#1{\mbox{\boldmath $#1$}}

\newcommand{\vgz}{\mbox{$\bm \zeta$}}

\newcommand{\trace}{\ensuremath{\hbox{tr}}}

\newtheorem{theorem}{Theorem}
\newtheorem{lemma}[theorem]{Lemma}
\newtheorem{prop}{Proposition}
\newtheorem{claim}[theorem]{Claim}

\newtheorem{definition}{Definition}
\newtheorem{remark}{Remark}

\newcommand{\beq}{\begin{equation}}
\newcommand{\eeq}{\end{equation}}
\newcommand{\bea}{\begin{array}}
\newcommand{\ena}{\end{array}}
\newcommand{\bds}{\begin {description}}
\newcommand{\eds}{\end {description}}
\newcommand{\bdf}{\begin{definition}}
\newcommand{\blm}{\begin{lemma}}
\newcommand{\edf}{\end{definition}}
\newcommand{\elm}{\end{lemma}}
\newcommand{\bthm}{\begin{theorem}}
\newcommand{\ethm}{\end{theorem}}
\newcommand{\bprp}{\begin{prop}}
\newcommand{\eprp}{\end{prop}}
\newcommand{\bcl}{\begin{claim}}
\newcommand{\ecl}{\end{claim}}
\newcommand{\bcr}{\begin{coro}}
\newcommand{\ecr}{\end{coro}}
\newcommand{\bquest}{\begin{question}}
\newcommand{\equest}{\end{question}}

\newcommand{\rarrow}{{\rightarrow}}

\begin{document}
% ===================================================================
% Document Start
% ===================================================================
\title{The  Interference Channel Revisited: Aligning Interference by Adjusting Receive Antenna Separation}
%Nulling an Arbitrary Number of Directions via Adaptive Antenna Spacing}
\author{Amir Leshem\thanks{Amir Leshem is with Faculty of Engineering, Bar-Ilan University (leshema@biu.ac.il). The work of Amir Leshem was partially supported by ISF grant 1644/18 and ISF-NRF grant  2277/16. Parts of this paper will be presented in IEEE ISIT 2019 \cite{leshem2019ergodic}.} and Uri Erez\thanks{U. Erez is with Tel Aviv University, Tel Aviv,
Israel (email: uri@eng.tau.ac.il). The work of U. Erez was supported by by the ISF under Grant 1956/15.}}
\maketitle

\begin{abstract}

It is shown that a receiver equipped with two antennas may null an arbitrary large number of spatial directions to any desired accuracy, while maintaining the interference-free signal-to-noise ratio, by judiciously adjusting the distance between the antenna elements.
The main theoretical result builds on ergodic theory. The practicality of the scheme in moderate signal-to-noise systems is demonstrated for a scenario where each transmitter is equipped with a single antenna and each receiver has two receive chains and where the desired spacing between antenna elements is achieved by selecting the appropriate antennas from a large linear antenna array.
%From a practical perspective, implementing the selection mechanism yields a substantial reduction in hardware complexity.
We further extend the proposed scheme to show that interference can be eliminated also in specular multipath channels as well as multiple-input multiple-output interference channels where a single extra receiver suffices to align all interferers into a  one-dimensional subspace.
To demonstrate the performance of the scheme, we
show significant gains for interference channels with four as well as six users, at low to moderate signal-to-noise ratios ($0-20$ dB). The robustness of the proposed technique to small channel estimation errors is also explored.

\end{abstract}

%=============================================================================

%
\section{Introduction}

The information-theoretic model of an interference channel is an abstraction that is motivated by the physical channel model of transmitter-receiver pairs that communicate over a shared  wireless medium. While abstraction often leads to  insights that may then be translated to more complicated real-life models, it is now recognized that the interference channel is an example that  generalization also carries with it the risk of over-abstraction, i.e., losing some key features of the true problem. Therefore,  it is worthwhile to re-examine the problem formulation from time to time as has been demonstrated, e.g., in the case of magnetic recording channels; see e.g., \cite{immink1998codes} for an overview of the  evolution of the physical models  and its impact on the relevant  information-theoretic and coding techniques. Another example is the evolution that led to the V.90 voice-band modem \cite{kim2004v,humblet1996information}.
Indeed, works on interference alignment  \cite{cadambe2008interference,motahari2014real,nazer2012ergodic} reveal that a two-user model is non-representative and further that the linear Gaussian model allows for elegant schemes that do not carry over to the general interference channel model.
% Nonetheless, interference alignment techniques have faced serious difficulties in translating the theoretical asymptotic gains to the operating conditions of communication systems.

In the present paper, we argue that taking a further step in bringing back into the model some simple considerations stemming from the physical propagation medium yields new insights on how interference may be nearly eliminated  by an effective signal processing method, thereby resolving some of the drawbacks of
existing interference alignment  techniques. Specifically, we prove that all interferers can me nearly aligned into a one-dimensional subspace that is nearly orthogonal to the direction of the desired user.
This property is crucial as it eliminates the power penalty involved in the interference alignment techniques of  \cite{cadambe2008interference,motahari2014real} which restrict the gains to the very high signal-to-noise ratio (SNR) regime. Further, a major advantage of the proposed method over existing interference alignment schemes \cite{cadambe2008interference,motahari2014real,nazer2012ergodic}  is that it relies on the availability of channel state information (CSI) at the receiver side only.
This property is of significant importance as CSI feedback to the transmitters is recognized
as a major hurdle for realizing the gains of interference alignment in practical systems; see, e.g., \cite{rao2013csi,chen2014performance,thukral2009interference}.

To convey  the essence  of the advocated approach, consider standard interference nulling as performed in multi-antenna wireless communication. It is well known that given an adaptive array with $N_r$ receive antennas, one can null out $N_r-1$ (single-antenna) interferers and enjoy a full degree-of-freedom (DoF) for one (single-antenna) desired source. This leads to low utilization of the receive antennas, since only $\nicefrac{1}{N_r}$ of the DoFs of a single-user MIMO system are attained.
Nonetheless, receive beamforming is very simple to implement and is robust, relying only on receiver-side CSI. Therefore, it is the workhorse of modern wireless communication systems.

Nevertheless, as mentioned above, works on achievable rates for the interference channel \cite{cadambe2008interference,motahari2014real,nazer2012ergodic} demonstrate that  half of the degrees of freedom can  be achieved, independent of the number of interferers, even when employing a single antenna at each node. While appealing from  a theoretical point of view, interference alignment techniques face some major challenges in real-life applications; see, e.g., \cite{el2013practical}. Beyond  knowledge of full CSI of the complete interference network at all transmitters being required, the results are highly asymptotic. The SNR at which a tangible improvement over naive schemes is achieved is extremely high or requires very specific system configurations, such as a symmetric  or cyclic interference channel models \cite{ordentlich2014approximate,zhao2016interference}  or sporadic low-dimensional MIMO  configurations  \cite{yetis2010feasibility}.

%Specifically, the proposed method substitutes simultaneous alignment of interference by simultaneous \emph{near nulling}. This difference is crucial as it eliminates the power penalty involved in interference alignment techniques and thus the associated gains are not restricted to the very high signal-to-noise ratio (SNR) regime.

Most of the work on the interference channel concentrated on a simplified channel model which assumes that the wireless channel is represented by an arbitrary matrix with random elements. However, as is recognized for many years in the communications theory literature, wireless channels are better represented as a combination of a small number of reflections with complex random coefficients caused by the small scale fading at the reflectors. Examples of such models include the well known  Saleh and Valenzuela model \cite{saleh1987statistical}  that is prevalent in recent applications of wireless communications  (see e.g., \cite{rappaport2013millimeter} and the references therein), as well as ray-based MIMO models \cite{almers2007survey,shiu2000fading,bolcskei2002capacity,oestges2003physical,xu2004generalized}. These models are characterized by a finite (typically small) set of reflection clusters with well defined direction-of-departure (DoD) and direction-of-arrival angles (DoA) together with fading coefficients. %where a LOS channel uses a Rician distribution and the non-LOS have Rayleigh fading coefficient.

We consider the particular class of line-of-sight (LOS) interference channels as well as specular multipath channels  that are the basis for physical modeling of recent of wireless communication systems. We further extend the results to the MIMO interference channel.
%based on physical channel models as well as multi-hop interference networks and cellular interference multiple access channels.

% Our main result is as follows: in a LOS planar interference channel, by properly setting the distance between the antennas of a two-antenna receiver, an arbitraity large numer of interferers may be aligned into a one-dimensional subsapce that is nealy orthogonal to that of the desired user.

Our main theoretical results may be translated to a practical implementation by using  a large $\gl$-spaced array with $N_r$ receive antennas and a selection mechanism as depicted in Figure~\ref{fig:selection}.
The main result may then be restated as: Selecting  {\em two} antennas at distance $d=m \lambda$, $m$ being an integer, out of a large linear array suffices to approximately  null out {\em any} number of sources in the plane with negligible noise amplification.

Accordingly, for a $K$-user interference channel with single-antenna transmitters, we achieve $K$ degrees of freedom, i.e., a utilization factor of one half, similar to the best achievable DoF of existing interference alignment schemes. In contrast to the latter, the proposed scheme  requires only receiver-side  CSI. Moreover, it achieves substantial gains  at practical values of SNR.
%The scheme can be implemented using a simple array and selection mechanism as depicted in Figure~\ref{fig:selection}.
\begin{figure}[htp]
\begin{centering}
\psfrag{A}{$t$}
\psfrag{B}{\tiny {$N_t$ \ \  }}
\psfrag{C}{\tiny $N_r$ \ \ }
\psfrag{D}{$r$}
\includegraphics[width=\columnwidth]{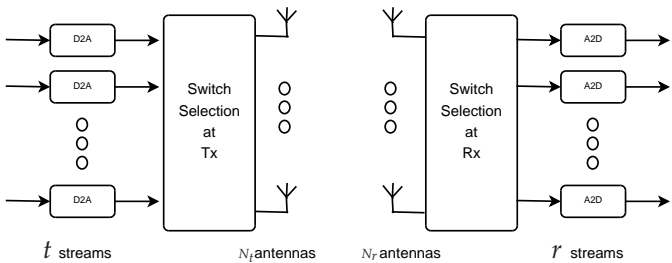}
\end{centering}
\caption{Setting antenna spacing via selection.}
\label{fig:selection}
\end{figure}

% It will be demonstrated that  the considered physical channel model yields very different achievable rates that those attained  with scaling that is larger than the degrees-of-freedom in the random algebraic channel model. Surprisingly, the degrees of freedom analysis proposed by existing IC literature is pessimistic. By selecting judiciously, one of $N_r$ IC's using antenna selection scheme, we are able to prove that for large enough $N_r$, interference can be eliminated to any desired degree with only two receivers, while receiving the desired signal with no power degradation. This yields not only capacity scaling, but also provides implementable schemes for various interference channel models, which is better than any former implementable scheme for the interference channel.
% In contrast, \cite{gomadam2011distributed} shows that for $3$ users interference channel with 2 transmitters and 2 receivers for each user a single degree of freedom is achieved. In contrast, our approach leads to the same rate when users have single transmitter and two receivers!

To gain insight into the proposed approach, consider a four-user LOS interference channel where we focus on the receiver of user 1. The direction of the transmitters are $[175^o, 59^o, 151^o, 133^o]$ with respect to the antenna array of user 1. Selecting antennas having a separation of $5\gl$, yields the beam pattern depicted in  Figure~\ref{fig:beampattern}.  The desired user's gain is close to 2 which is the interference-free gain, while the signals of all other users are almost completely suppressed. Theorem~\ref{thm:main} proves that a beam pattern satisfying this property almost always exists, for some pair of antennas,  provided that the array is sufficiently large.

\begin{figure}[htbp]
\centering
\includegraphics[width=0.55\columnwidth]{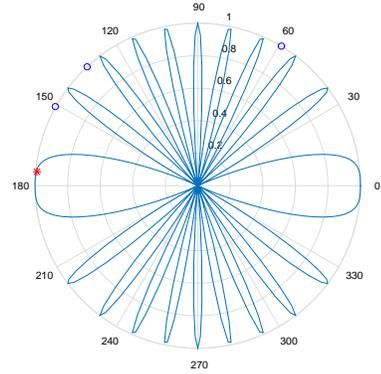}
\caption{Optimal beam pattern of user 1 in a four-user interference channel where the DoA with respect to the receive array of user 1 are $[175^o, 59^o, 151^o, 133^o]$. The optimal antenna separation is $d=5\gl$ for an array with $d_{\max}=25 \gl$. The powers of  all signals are assumed to be equal.  %$P=1$.
}
\label{fig:beampattern}
\end{figure}
\subsection{Further related work}

Both the Cadambe-Jafar and the Motahari et al. techniques are very asymptotic in nature and require high resolution transmit-side CSI as well as very high SNR conditions to start to play a beneficial role.
Extensions to more general MIMO channels have subsequently revealed, e.g., \cite{bresler2014feasibility, yetis2010feasibility}  that the DoF alignment gains are much more modest under more realistic assumptions.
As a partial remedy, antenna switching has been proposed as a means for improving the channel coefficients to facilitate alignment \cite{gou2011aiming}.

Apart from the obvious connection to  works on the interference channel, the idea of altering the physical propagation channel bears some similarity to ``media-based modulation", ``spatial modulation" and ``index modulation" schemes; see \cite{khandani2014media, naresh2017media, basar2017index,ishikawa201850} for an overview of these inter-related concepts.
In all of these works, the physical medium is \emph{modulated} based on the information-bearing signal.  In contrast to these works, the present work only requires sub-sampling of the spatial channel at the receiver and can cope with an arbitrary number of interfering signals.

The rest of this paper is organized as follows. Section~\ref{sec:2} gives a description of the MIMO interference channel model where both transmitter and receiver can only utilize (i.e., must select) a prescribed number of antenna elements. Section~\ref{sec:LOS_int_channel} specializes the model to a LOS scenario with single-antenna transmitters.
The main result of the paper is presented in Section~\ref{sec:4}
in the context of LOS channels. It is extended to specular multipath as well as MIMO scenarios in Sections~\ref{sec:multipath} and \ref{sec:MIMO}. Optimization algorithms for antenna selection and simulation results are then given in Section~\ref{LOS_opt} and Section~\ref{sec:sim}, respectively. The paper concludes with a discussion of possible extensions in Section~\ref{sec:disc}.

\section{Interference Channel with Selection}
\label{sec:2}
We extend the standard interference channel model to a model where
the distance between receive antennas may be adjusted. We will focus and describe such a system via antenna selection  applied to a large linear antenna array.
In such a system, the number of transmit/receive chains for each user is not necessarily equal to the number of transmit/receive antennas, respectively.

Thus, we consider a $K$-user interference channel where each transmitter has $t$ transmit chains and $N_t$ antennas and each receiver has $r$ receive chains and $N_r$ antennas;  see Figure~\ref{fig:selection} depicting a link between one transmitter and one receiver. This configuration, where all transmitter/receiver pairs have the same parameters, is denoted as the symmetric $t/N_t/N_r/r$ interference channel; nonetheless, the results carry over to  non-symmetric configurations.
We note that  this model is of significance to modern systems that utilize massive antenna arrays.
%where the number of Tx/Rx chains differs from the number of antennas.

% Consider an interference channel with $K$ transmitters and $K$ corresponding receivers. We assume for simplicity that all transmitters are equipped with the same number of antennas $N_t$ and all receivers are equipped with $N_r$ antennas.
We now formalize the channel model, beginning with the traditional case where $t=N_t$ and $r=N_r$.
Denoting the $N_r \times N_t$ channel matrix from transmitter $j$ to receiver $i$ by $\mH_{ij}$, the  received signal is given by
\begin{align}
\label{eq:IC_model}
    \mathbf{y}_i=\sum_{j=1}^K \mathbf{H}_{ij} \mathbf{x}_j + \mathbf{z}_i, \quad i=1\ldots K,
\end{align}
where $\mathbf{z}_i$ is i.i.d. (between users and over time) circularly-symmetric complex Gaussian noise with variance $\sigma^2$ per complex dimension.

% Several variants of this problem have been addressed.
% For instance, the case of  $N_t=N_r=1$ and real time-varying (which can be thought of as a diagonal matrix) coefficients has been studied in \cite{cadambe2008interference} where it was shown that for almost all channel coefficients, interference alignment attains half a DoF per user.
% A similar result was shown for scalar but time-invariant channels in \cite{motahari2014real} through alignment on the signal scale using lattice codes.

% Both of these approaches are very asymptotic in nature and require high resolution transmit-side CSI as well as very high SNR conditions to start to play a beneficial role.
% Extensions to more general MIMO channels have subsequently revealed, e.g., \cite{bresler2014feasibility, yetis2010feasibility}  that the DoF alignment gains are much more modest under more realistic assumptions.
% As a partial remedy, antenna switching has been proposed as a means for improving the channel coefficients to facilitate alignment \cite{gou2011aiming}.
The $t/N_t/N_r/r$-interference channel can be described by  requiring that  for each user $i$, the transmitter and receiver must employ linear front-end selection matrices $\mS_T\in \{0,1\}^{N_t \times t},\mS_R \in \{0,1\}^{N_r \times r} $, each having exactly $t$ and $r$ non-zero elements that are in the same row or column, respectively. Applying selection matrices $\mS_{T,i},\mS_{R,i}$  at both ends of the link of  each user, \eqref{eq:IC_model} becomes
\begin{align}
\label{eq:IC_model2}
    \mathbf{y}_i&=\sum_{j=1}^K \mS_{R,i}^H \mathbf{H}_{ij} \mS_{T,j} \mathbf{x}_j + \mathbf{z}_i, \nonumber  \\
     &=\sum_{j=1}^K \mathbf{H}^S_{ij} \mathbf{x}_j + \mathbf{z}_i.
\end{align}

We note that the $t/N_t/N_r/r$-interference channel model may be formalized in an information-theoretic framework as an input/output constrained channel. Specifically, the input constraint amounts to a imposing the constraint that the encoding function $f$ must satisfy:
 $$m\in \left\{1,\ldots,2^{nR}\right\} \stackrel{f}{\rightarrow} \left(\mathbf{x}_1(m),\ldots,\mathbf{x}_n(m)
 %;\mathbf{s}_t
 \right),$$
%Here, $\mathbf{s}_t\in\{0,1\}^{N_t}$ is a selection vector %(constant for all symbols $n$ within a code block) satisfying %the support constraint $\|\mathbf{s}_t\|_0\leq t$ and %$\mathbf{x}_i \in \mathbb{C}^{N_t}$ is the vector input at time %$i=1,\ldots,n$.
%The transmitted signals
%%In a practical coding scenario,
%In a practical realization,
where the codeword entries $\vx_i(m)$ must be $t$-sparse, with the same sparsity pattern for all time instances $i$.
%with the non-zero elements placed in positions (antennas) corresponding to the non-zero elements of $\mathbf{s}_t$.

Similarly, the output constraint can be formalized as a mismatched decoding constraint requiring that the decoding function $g$
$$\left(\mathbf{y}_1,\ldots,\mathbf{y}_n\right)\stackrel{g}{\rightarrow} m\in \left\{1,\ldots,2^{nR}\right\},$$
has the form $g\left(\mathbf{y}_1,\ldots,\mathbf{y}_n\right)=g'(\vgz_1,\ldots,\vgz_n)$  where
$\vgz_i=\mathbf{s}_r \odot \mathbf{y}_i$ ($\odot$ denoting component-wise multiplication) and  where $\mathbf{s}_r$
must satisfy the support constraint $\|\mathbf{s}_r\|_0 \leq r$.

\section{Line-of-sight interference channels}
\label{sec:LOS_int_channel}
The use of high-frequency communication in general and mm-wave and THz frequency communication has prompted recent interest in LOS communication channels \cite{witrisal2016high, heath2016overview, boccardi2014five}. Moreover, such channels form the basis for the more elaborate channel models described in Sections \ref{sec:multipath} and \ref{sec:MIMO}.

We now describe the $1/1/N_r/2$-LOS interference channel. The antenna array is assumed to have $N_r$ uniformly $\lambda$-spaced elements, out of which two are switched into the receive chains.
We assume that the receiver has full directional CSI.

We make the following assumptions and comments:
\begin{enumerate}[label=A\arabic*]
    \item For the $1/1/N_r/2$ channel, $S_{T,i}$ is trivial for all $i$.
    \item For simplicity, we assume a linear array and planar geometry where all sources are far-field point sources.
    \item Let $\vh(\theta;d)$ be the array response towards direction $\theta$ with separation $d$ between the two selected antennas. Thus, the array response is given by
 \begin{equation}
 \label{def:h}
    \vh(\gth;d)=\frac{1}{\sqrt{2}}\left[1, e^{j 2\pi d \cos \theta} \right]^T.
\end{equation}
    \item The vectors $\vh(\theta_{i,j};d)$  consist of array manifold vectors for signals impinging on the array of receiver $i$ from direction $\theta_{i,j}$.
    %which depend on the selection matrix, setting the spacing $d$ (in units of $\lambda$) between the chosen antennas.
   % \item We allow the spacing between the two receive antennas to be altered as needed to optimize performance.
  \item We assume two receive chains per user, i.e. the matrices $\mH^S_{i,j}$ in \eqref{eq:IC_model2} are reduced to $2 \times 1$ LOS vectors $\vh_{i,j}=\vh(\gth_{i,j};d)$.
  \item We  assume that each receiver has perfect CSI w.r.t. all channel gains corresponding to impinging signals.
  \item Transmitters on the other hand need not have access to any CSI beyond the rate at which they should communicate with their respective receiver.
  \item Without loss of generality, we use the array manifold as the channel, since the signal attenuation can be absorbed in the power of the signal $x_j$.
  \item \label{independence} We assume that the locations of all  transmitters and receivers are independently  uniformly distributed in angle with respect to the origin.
  \item We assume that the power of all transmitters is bounded by $P$.
\end{enumerate}
Note that by \ref{independence}, the incidence angle of each received signal is uniformly distributed as well.
Under this setting, it suffices to focus on the receiver of a single user $i$  as the operations at all receivers will be similar.
% From this point on, we therefore drop the user index $i$.

It follows that \eqref{eq:IC_model} becomes:
\begin{equation}
    \mathbf{y}_i=\sum_{j=1}^K \vh(\gth_{i,j};d_i) \mathbf{x}_j + \mathbf{z}_i, \quad i=1\ldots K,
%    \label{eq:IC_model5a}
\end{equation}
Using a received beamforming vector $\vw_i=\frac{1}{\sqrt{2}}[1, e^{j\phi_i}]$, the received signal of the $i'$th user becomes:
\begin{equation}
\label{eq:rx_bf}
    y_i=\sum_{j=1}^K \kappa(\gth_{i,j};d_i) \mathbf{x}_j + z_i,
\end{equation}
where
\begin{align}
\label{def:kappa}
    \kappa(\gth_{i,j};d_i)=\vw_i^T\vh(\gth_{i,j};d_i).
\end{align}

% Consider a single receiver with a desired signal impinging from direction $\gth_i$. Let $\gth_j:j\neq i $ be the directions of the interfering signals.
Therefore,
  \begin{align}
      g(\gth_{i,j};d_i)=\left|\kappa(\gth_{i,j};d_i)\right|^2= \frac{1}{{2}}\left|1+e^{j (2\pi d_i cos \theta_{i,j}+\phi_i)}\right|^2.
\end{align}
Straightforward algebraic simplification yields:
\begin{equation}
     g(\gth_{i,j};d_i)=\frac{1}{2}\left[1+\cos(2\pi d_i \cos(\theta_{i,j})+\phi_i)\right].
     \label{eq:thetaij}
\end{equation}

 In the next section we show that by properly selecting $d_i, \phi_i$, we can obtain the following:
\begin{equation}
\label{eq:design_eq}
 g(\gth_{i,j};d_i) \approx \gd_{i,j}, \quad j=1,..,K,
\end{equation}
where $\gd_{i,j}$ is Kronecker's delta function.

% Henceforth, we will omit the user index $i$, since the design of $\vw_i$ is done independently for each user.

% \begin{figure}[t]
% \label{fig:ambiguous}
% \begin{center}
% \includegraphics[width=0.7\columnwidth]{figures/ambiguous_array.eps}
% \caption{Beam pattern of an ambiguous array.}
% \end{center}
% \end{figure}

\section{Eliminating Interference via Ergodic Nulling}
\label{sec:4}
The classical signal processing literature deals primarily with Nyquist-resolution beamformers, where at least some antennas are separated by at most $\gl/2$. In this case, the array has a single main lobe in the desired direction, and the resolution of the array is determined   by the farthermost elements. This is  because  that when all distances between antennas are larger than   $\gl/2$,  an ambiguous beam pattern occurs. An example of this phenomenon is depicted in Figure~\ref{fig:beampattern}.
Interestingly, an ambiguous beam pattern, can prove extremely advantageous when dealing with interference, since such patterns have multiple nulls.
%We now show that by judiciously designing the beam pattern, we can point multiple nulls at the multiple interferers simultaneously.
%   In fact, with a highly under-sampled array, any number of interferers in almost any set of directions can be suppressed. This follows from an ergodic theory argument.

We now show that by judiciously adjusting (via selection) the distance between the receive antennas,  we can (with probability 1) suppress all interferers to any desired level. %Furthermore, we show that this can be achieved by a separation which is an integer multiple of the wavelength and hence, implementation via antenna selection applied to a linear array is possible.
% First, we leverage the theory of Diophantine approximations to establish that we may do so while treating the angle and gain of the desired signal as a random variable.
% We then significantly strengthen the results by proving
Specifically, we demonstrate that for almost all angles of arrival, one can approach the interference-free rate of any desired user, i.e., the rate achievable with a two-antenna receiver when no interference is present. This is proved using the uniform distribution property of  sequences modulo $1$.

Denote $\vwt$ the effective beamforming vector after applying the selection matrix, i.e.,
$\vwt_i=\mS_r \vw$, where $\vw$ is an $N_r$-dimensional beamforming vector.
We show that that the beamforming vector $\vw_i$ can be chosen as
\[
\vwt_i=\frac{1}{\sqrt 2} [1,1]^T = \mS_r \vone,
\]
where $\vone$ is the all-ones vector of length $N_r$.
The proof utilizes an integer  antenna spacing (in terms of wavelength $\lambda$).

\begin{theorem}[Main Theorem]
\label{thm:main}
Assume that the directions $\gth_{i,1},..,\gth_{i,K}$ are such that $\cos(\gth_{i,1}),...,\cos(\gth_{i,K})$ are independent over $\mathbb Q$. Then, for  every $\gd>0$, one can find a spacing $d \in \mathbb N$ such that applying receive beamforming with the vector $\vwt_i=\frac{1}{\sqrt 2} [1,1]^T$  yields:

\beq
 \label{eq:array_gain}
\begin{array}{lcll}
g(\gth_{i,k};d_i)&<& \gd,  \quad \quad k \neq i\\
g(\gth_{i,i};d_i)&>&1-\gd.
\end{array}
\eeq
\end{theorem}
Note that Theorem~\ref{thm:main} implies that all interference directions are aligned to a (complex) one-dimensional subspace, nearly orthogonal to the array manifold of (the desired) user $i$.
Not only does this provide a full DoF per transmitter but also no noise amplification occurs due to the near-orthogonality of interference and desired signals.
% . This is the case since the gain in the desired direction can be made arbitrarily close to 1 while the total interference is suppressed to any desired level.

\begin{proof}
To prove the main theorem, recall the following definition by Weyl  (see \cite{kuipers2012uniform}):
\bdf
A $K$-dimensional sequence of real vectors $\vx_m:m\in {\mathbb N}$ is uniformly distributed modulo $1$ if for every box
\[
B=\prod_{k=1}^K[a_k,b_k],  \quad B \subseteq [0,1)^K
\]
\begin{align}
&\lim_{M \rarrow \infty}\frac{\left|\left\{1\le m \le M : \left(\vx_m \mod 1 \right)\in B \right \}\right|}{M} \nonumber \\
&=\prod_{k=1}^K(b_k-a_k)
\end{align}
\edf
Weyl \cite{weyl1916gleichverteilung}  proved that whenever $\vx=[x_1,...,x_K]^T$ is a vector of irrational  real numbers that are linearly independent over $\mathbb Q$, the sequence
$\left\{m\vx \mod 1 :m \in \mathbb N \right\}$ is uniformly distributed modulo $1$. In the present context, assume that $\cos(\gth_{i,1}),...,\cos(\gth_{i,K})$ are linearly independent over $\mathbb{Q}$. Note that this holds with probability one.
%Let
%\[
%\vc=[\cos(\gth_{i,1}),...,\cos(\gth_{%i,K})].
%\]
By Weyl's theorem the sequence $m [\cos(\gth_{i,1}),...,\cos(\gth_{i,K})]$, \mbox{$m \in \mathbb{N}$}, is uniformly distributed  modulo 1. Define a box
\beq
B=\prod_{k=1}^K B_k
\eeq
where
\[
B_k= \left\{
\begin{array}{cl}
[0, \gre'] & k=i  \\
\left[\frac{1-\gre'}{2},\frac{1+\gre'}{2}\right] & k\neq i  \\
\end{array}
\right.
\]
where $\gre'=\frac{\gre}{2 \pi}$
Therefore, we can find a $d$ such that
\begin{eqnarray}
2 \pi d \cos(\gth_{i,i}) <\gre \mod 2\pi \\
\pi-{\frac{\gre}{2}<2 \pi d \cos(\gth_{i,j}) }<\pi+\frac{\gre}{2} \mod 2\pi
\end{eqnarray}
By continuity of $g(\gth_{i,j};d_i)$, as given in \eqref{eq:thetaij} , for a given $\gd$, we can find an $\gre$ such that \eqref{eq:array_gain} is satisfied.
\end{proof}

\section{Multipath Channels and related interference channel models}
\label{sec:multipath}
We now show that the proposed approach generalizes to the case of multipath \cite{tseviswanath,gallager2008principles} with  a finite number of reflections.

We assume a physical channel model such that the carrier frequency is much larger than the signal bandwidth, which is typical in cellular and indoor wireless communications.
We first show that in the case of a $1/1/N_r/2$ interference channel with inter-symbol interference (ISI), we can approach  the single-user (with same number of receive chains) interference-free rate.

% Under the above assumptions the channel between the transmitters of all users and the receive array of user $i$ can be described in the frequency domain by:

Adhering to discrete time and allowing a different path loss for each reflection, the channel, as given in \eqref{eq:IC_model2},  now generalizes to
% \begin{align}
%     \mathbf{y}_i(\ebj)=\sum_{k=1}^K \sum_{\ell=1}^{L_{i,j}}  \grg_{i,k,\ell}\mathbf{h}(\gth_{i,k, \ell}) e^{j \omega \tau_{i,k, \ell}}  \mathbf{x}_k(\ebj) + \mathbf{z}_i(\ebj),
%     \label{eq:IC_model3}
% \end{align}
\begin{align}
    \mathbf{y}_i(t)=\sum_{k=1}^K \sum_{\ell=1}^{L_{i,j}}  \grg_{i,k,\ell}\mathbf{h}(\gth_{i,k, \ell};d_i)  \mathbf{x}_k(t-\tau_{i,k, \ell}) + \mathbf{z}_i(t),
    \label{eq:IC_model3}
\end{align}
for $i=1,\ldots,K$, where $L_{i,j}$ is the number of reflections of the $j$'th signal received by user $i$, and $\grg_{i,k,\ell}$ is the complex path loss of the signal arriving from direction $\gth_{i,k, \ell}$ and $\vh$ is defined in \eqref{def:h} incorporating the local scattering at the reflector.  Following standard models, we can assume that  $\grg_{i,k,\ell}$  is either stochastic, e.g., Rician or Rayleigh, or deterministic.

Let
\beq
B_i=\prod_{k=1}^K \prod_{\ell=1}^{L_{i,k}} B_{i,k, \ell}
\eeq
where for all $1 \le \ell \le L_{i,k}$:
\[
B_{i,k, \ell}= \left\{
\begin{array}{cl}
[0, \gre'] & k=i, \ell=1,...,L_{i,i}  \\
\left[\frac{1-\gre'}{2},\frac{1+\gre'}{2}\right] & k\neq i, \ell=1,...,L_{i,k}  \\
\end{array}
\right.
\]
and $\gre'=\frac{\gre}{2 \pi}$. As in the previous section, by invoking Weyl's theorem,  there exists a $d$ such that:
\begin{eqnarray}
2 \pi d \cos(\gth_{i,i,\ell}) <\gre \mod 2\pi
\end{eqnarray}
for all $\ell=1,\ldots,L_{i,i}$ and
\begin{eqnarray}
\pi-{\frac{\gre}{2}<2 \pi d \cos(\gth_{i,j,\ell}) }<\pi+\frac{\gre}{2} \mod 2\pi
\end{eqnarray}
for all $\ell=1,\ldots,L_{i,j}$ and $j \neq i$.
By continuity of $g(\gth;d)$, for a given $\gd$, there exists an $\gre$ and a $d_i=d$ such that:
\begin{align}
    g(\gth_{i,i,\ell};d_i)&>1-\gd, \quad \ell=1,...,L_{i,i}\\
     g(\gth_{i,j,\ell};d_i)&<\gd, \quad \ell=1,...,L_{i,j}, j \neq i.
\end{align}

We conclude that one can suppress all specular multipath components of the interfering signals to any desired level. Hence, the resulting received signal is given by:
\begin{equation}
   \mathbf{y}_i(t)=\sum_{\ell=1}^{L_{i,i}}\kappa(\gth_{i,i,\ell};d_i) \grg_{i,i,\ell} \mathbf{x}_i(t-\tau_{i,i, \ell})+\tilde{\vz}_i(t)
\label{eq:IC_model4}
\end{equation}
where
\[
\tilde{\vz}_i(t)=\vz_i(t)+\vz'_i(t)
\]
is composed of the receiver noise as well as the  residual interference at receiver $i$, $$\vz'_i(t)=\sum_{k \neq i}\sum_{\ell=1}^{L_{i,k}}\kappa(\gth_{i,k,\ell};d_i) \grg_{i,k,\ell}  \mathbf{x}_k(t-\tau_{i,k, \ell}).$$
Note that the power of the residual interference satisfies:
\begin{equation}
    E\left\| \vz'_i(t)\right\|^2<\gd  \sum_{k \neq i} E\left|\vx_k(t) \right|^2 \sum_{\ell=1}^{L_{i,k}} |\grg_{i,k,\ell}|^2
\end{equation}
By selecting $\gd$ sufficiently small,  $E|\vz'_i(t)|^2$ can  be
 can be made arbitrarily small. Moreover, for all desired signal paths $g(\gth_{i,i,\ell};d_i)=\left|\kappa(\gth_{i,i,\ell};d_i) \right|^2 $ are (simultaneously) arbitrarily close to $1$ by a proper choice of $\gd$. It follows that \eqref{eq:IC_model4} amounts to a standard ISI channel, with coefficients arbitrarily close to the interference-free ISI channel.
\section{Ergodic nulling for the MIMO interference channel}
\label{sec:MIMO}
We now turn to analyze the MIMO interference channel where for simplicity we assume that the number of transmit and receive antenna elements as well as RF chains is the same for all transmitter and receiver pairs, i.e., of dimensions $N_t$, $N_r$, $t$ and $r$.
%between  $N_t\times N_r$  MIMO channe

%using the same physical channel model as above.
Following the vast literature of physical spatial point-to-point MIMO channel models, we note that the $N_r\times N_t$  MIMO channel between the transmit antennas of user $j$ and the receive antennas of user $i$ can be described as
\begin{equation}
\label{def:H_mimo}
    \mH_{i,j}=\sum_{\ell=1}^{L_{i,j}} \grg_{i,j,\ell} \va_{R,i}(\gth_{i,i,\ell}) \va_{T,j}(\psi_{i,j,\ell})^T
\end{equation}
where, $\psi_{j,\ell},\gth_{i,\ell}$ are the DoD between transmit array $j$ and reflector $\ell$ and the DoA between receive array $i$ and reflection $\ell$. Without loss of generality we also assume that $|\grg_{i,j,\ell}|$ is monotonically decreasing in $\ell$.
As is common in the MIMO literature, we assume that the scattering is sufficiently rich. In the present context, this requires that $L_{i,i} \geq t$ for all $i$ so that (almost surely) for all $i$, we have $\hbox{rank}(\mH_{i,i}) \geq t$.
%, i.e., the directions $\gth_{i,j,l}$ satisfy that $\cos(\gth_{i,j,l})$ are irrational and independent over the rationals. This holds with probability 1.
The following theorem holds:
\begin{theorem}
\label{thm2}
Let $t,N_t$ be given and assume that $r=t+1$, $L_{i,i} \geq t$ and $N_t=t$. Further, assume that each receiver has directional CSI. Then, for any  $\delta>0$, there is a sufficiently large $N_r$ and a selection matrix $\mS_{R,i}$, such that for user $i$ any rate satisfying
\begin{equation}
R_i \le \log \left|\mI+\frac{P}{\gs^2 t} \mG_i \mA_{T,i} \mA_{T,i}^H \mG_i^H \right|-\gd
\end{equation}
is achievable in the  $t/N_t/N_r/(t+1)$ interference channel,
where
\[
\mG_i={\rm diag}\{\grg_{i,i,1},...,\grg_{i,i,t}\}
\]
and
\[
\mA_T=\left[\va_{T,i}(\psi_{i,i,1}),\ldots,\va_{T,i}(\psi_{i,i,t}) \right]
\]
Furthermore, if the transmitter has CSI,
%and knows also $\grg_{i,i,\ell}: \ell=1,...,t$
then any rate
\[
R_i \leq \max_{\mQ} \log \left|\mI+ \mG_i \mA_{T,i} \mQ \mA_{T,i}^H \mG_i^H \right|-\gd
\]
is achievable where $\mQ$ is a positive semi-definite matrix satisfying $\trace(\mQ)=P$.
%we can allocate power $P_{\ell}$ to %each direction where
%$P_{\ell}$ is determined by water-filling with respect to $\frac{\gs^2}{|\grg_{i,i,\ell}|^2\|\va_{T,i}(\psi_{i,i,\ell})\|^2}$.
\end{theorem}
\begin{proof}
Let the transmitter use an i.i.d. isoptropic Gaussian codebook of dimension $t$ and power $\nicefrac{P}{t}$ per dimension.
Let $\ve_i$ denote the  $(t+1)-$dimensional standard unit vectors.
The receiver uses a selection matrix $\mS_{R,i}$ followed by a beamforming matrix
$\mW_i=\left[\vw_1,...,\vw_t \right]$. This effectively translates to requiring that the vectors $\vw_i$ satisfy $\|\vw_i\|_0=2$. Furthermore, we can choose
\[
\vw_{i,\ell}=\frac{1}{\sqrt{2}}\left(\ve_0+\ve_{n_{i,\ell}}\right).
\]
Here, $(\mS_{R,i})_{\ell,n_{i,\ell}}=1$ if and only if the antenna $n_{i,\ell}$ is selected such that the beamformer $\vw_{i,\ell}$ receives only direction $\gth_{i,i,\ell}$ (and approximately nulling all other directions, both from the desired user as well as from all others users). Recall that by  Theorem~\ref{thm:main}, this is possible.

Thus, user $i$ obtains an equivalent MIMO channel
\begin{equation}
   \mathbf{y}_i= \mHt_i \mathbf{x}_i+\tilde{\vz}_i
\label{eq:IC_MIMO}
\end{equation}
where
\begin{align}
\label{eq:MIMO_channel_matrix}
\mHt_i=\mW_i\mS_{R,i}\mA_{R,i} \mG_i \mA_{T,i}
\end{align}
and
\[
\tilde{\vz}_i=\vz_i+\vz'_i
\]
is composed of the receiver noise as well as the  residual interference at receiver $i$.
%That is,
%$$\vz'_i=\sum_{k \neq %i}\sum_{\ell=1}^{L_{i,k}}g(\gth_{i,k,%\ell})   \mathbf{x}_k.$$
Note that the power of the residual interference can be made as small as desired.
%satisfies:
%\begin{equation}
%    \left\| \vz'_i(t)\right\|^2<\gd  %\sum_{k \neq i} %\sum_{\ell=1}^{L_{i,k}} %|\grg_{i,k,\ell}|^2
%\end{equation}
%By selecting $\gd$ sufficiently %small,  $|\vz_i(t)|^2$ can also be
% can be made arbitrarily small.

By construction
\begin{align}
    \mW_i\mS_{R,i}\mA_{R,i} = \mI+\mD_i
\label{eq:MIMO_equivalent_channel}
\end{align}
and $\|\mD_{i}\|_{\infty}<\gd$.
Hence, $\mHt_i$ can be made arbitrarily close to the channel
$\mHt'_i=\mG_i \mA_{T,i}$.

Thus, an achievable rate for this channel is given by
\[
R(\mHt'_i)=\log\left|\mI+\frac{P}{t\gs^2}  \mG_i \mA_{T,i}   \mA_{T,i}^H \mG_i^H \right|.
\]
The case of full CSIT follows by standard MIMO techniques.
\end{proof}
This should be compared against the isotropic transmission interference-free  benchmark with orthogonal channels
\[
\bar{R}(\mHt'_i)=\sum_{\ell=1}^{t+1}\log\left(1+|\grg_{i,\ell}|^2\frac{P}{\gs^2} \right).
\]
For large $t$, the two rates nearly coincide.

\begin{remark}
When $L_{i,i}>t$ we loose the low energy paths which are not included in the receive  beamformer. In this case, higher performance can be achieved by receiving the energy of the residual paths by increasing the number of receive chains up to $L_{i,i}$.  However, in realistic channel model, the number of dominant paths is relatively small.
\end{remark}
{\em Example: Three-user $2/2/N_r/3$ symmetric MIMO interference channel.}
Assuming a specular multipath model with finitely  many reflections  for every user and at least two for each desired user  (at the respective receiver), we note that we can achieve a total of $2$ DoFs per user. This is larger than the $\nicefrac{3}{2}$ DoFs per user achieved (in the generic) MIMO interference channel \cite{cadambe2008interference,bresler2014feasibility,yetis2010feasibility}. Moreover, the proposed scheme employs a much simpler transmission scheme, which does not require any CSI at the transmitter.

\section{Optimizing the receiver for a given array}
 Similarly to Theorem \ref{thm:main},
Theorem~\ref{thm2}  guarantees that interference can be suppressed to any desired level over finite multipath channels. However, it does not exploit the full optimization parameter space. Ultimately, our goal is to maximize the signal-to-interference-plus-noise ratio by properly choosing the antennas and the  beamformers corresponding to each source.
The straightforward approach would be to enumerate over all subsets of $r$ antennas and evaluating the SINR, for the optimal linear receiver. The complexity of this algorithm is prohibitive and simpler algorithms are called for.

 We begin by proposing a simple algorithm for antenna selection and  beamformer design for the LOS $1/1/N_r/r$-interference channel. Then we generalize the technique to the scenario of  MIMO multipath channels.
 \subsection{Pairwise antenna selection for LOS channels}
 \label{LOS_opt}
Consider the  LOS $1/1/N_r/r$-interference channel.  We wish to select  $r$ antennas as well as the beamforming vector. Optimizing over both involves a combinatorial search over all possible subsets of $r$ antennas and  computing the SINR attained by the optimal linear receiver for each subset. The total complexity of this search is $O\left(r^3 (N_r)^r \right)$.
To reduce the computational complexity, we propose a simple sub-optimal technique based on Theorem \ref{thm:main} to directly  select the antennas. Choose antenna $0$ as a reference. For each user, we search over the antennas and evaluate the two-antenna beamforming vector that maximizes the SINR  for this user. This amounts to computing
 \begin{align}
    \vw_i&=\arg \max_{\vw\in {\mathbb C}^{r}} \frac{P_i|\vw^H \va(\gth_i)|^2}
    {\sum_{j \neq i} P_j|\vw^H \va(\gth_j)|^2+\gs^2 \|\vw\|^2}  \nonumber \\
    {\rm subject\ to}:& \quad  \|\vw\|_
    0=2 \nonumber \\
                     &  \quad  w_0 =1.
\label{def:SINR_opt}
\end{align}

For each user, having chosen the antennas as above, we may further improve the combining weights applied to the chosen antennas.

To that end, let $n_0,...,n_{r-1}$ be the indices of the selected antennas and let $\mS_R$ be the corresponding selection matrix defined by $\left(n_0,...,n_{r-1}\right)$.
We can now maximize the SINR of each user by optimizing
\begin{align}
    \vw_i&=\arg \max_{\vw\in {\mathbb C}^{N_r}} \frac{P_i|\vw^H  \va(\gth_i)|^2}
    {\sum_{j \neq i} P_j|\vw^H \va(\gth_j)|^2+\gs^2 \|\vw\|^2} \nonumber\\
     {\rm subject\ to}:& \quad  {\rm supp}(\vw_i)=\left(n_0,...,n_{r-1}\right),
\end{align}
where ${\rm supp}$ defines the support.
Since the directions of interferers are assumed known, we can use the interference covariance-based beamformer \cite{gu2012robust} where the support constraint is incorporated by the selection matrix $\mS_R$:
\begin{equation}
    \vw_i=\mR_n^{-1} \mS^H_R\va(\gth_i)
\end{equation}
\begin{equation}
  \mR_n=\sum_{j \neq i} P_j \mS_R^H \va(\gth_j)\va(\gth_j)^H\mS_R+\gs^2 \mI.
\end{equation}
As discussed in \cite{ehrenberg2010sensitivity}, there is significant benefit in terms of robustness when using the interference covariance as a basis for beamforming instead of the received signal covariance matrix.

% \subsection{Pairwise antenna selection for finite multipath channels}
%  We begin by proposing a simple antenna selection and  beamformer design for the LOS $1/1/N_r/r$-interference channel for array with $N_r$ antennas. Then we generalize the technique to the scenario of multipath channels and MIMO systems.
 \subsection{Antenna selection and beamforming for finite multipath  MIMO channels}

The symmetric MIMO $t/t/N_r/r$-MIMO interference channel case can be treated similarly to the  specular multipath interference channel with some changes to the receiver structure. Assume that the MIMO channel is sufficiently rich so that the total number of multipath components between each transmitter and its respective receiver is larger than the number of spatial streams $t$. We focus on the case where CSI is not available at the transmitter (beyond agreed upon transmission rate) and thuse  we assume isotropic transmission,
%, where each stream is transmitted through a different antenna.
Recall that each path in the MIMO channel is described by the propagation matrix \eqref{def:H_mimo}.
The first phase of optimization selects a reference antenna and a single  antenna for each desired spatial reflection $\ell$ and a beamforming vector with two non-zero elements leading to \eqref{eq:MIMO_equivalent_channel}.

We assume a narrowband signal model, where the bandwidth of the transmitted signal is significantly smaller than the carrier frequency. Thus, the delays translate into phases. While applying Theorem~\ref{thm2} directly requires $L_{i,i}$ receive chains, we consider the equivalent multipath channel for each transmitted signal $x_m$, $m=1,\ldots,t$ defined by:

\begin{align}
\label{def:effective_channel}
    \vh_{i,j,m}=\sum_{\ell=1}^{L_{i,j}}\grg_{i,j,\ell} e^{j 2\pi f_c \gt_{i,j,\ell}} \va(\gth_{i,j,\ell})a_m(\psi_{i,j,\ell}),
\end{align}
where $f_c$ is the carrier frequency, and $a_m$ is the $m$'th entry of $\va$.
Hence, the signal to interfering users plus noise ratio for signal $m$  at receiver $i$ is given by
% \begin{align}
%     \mJ_{i,k}=\sum_{\ell=1}^{L_{i,k}}\grg_{i,k,\ell} e^{j 2\pi f_c \gt_{i,k,\ell}} \va(\gth_{i,k,\ell})\va(\psi_{i,k,\ell})^T.
% \end{align}
% % note that in the absence of interference, the optimal receiver amount to applying  Maximum Ration Combining (MRC) vector.

% Hence, we would like to choose two antennas (one is the refernce) maximizing the mutual information
%
\begin{align}
   \vw_m = \arg \max_{\vw\in {\mathbb C}^{r}} \frac{P_i|\vw^H \vh_{i,i,m}|^2}
    { I_{i,m}+\gs^2 \|\vw\|^2}  \nonumber \\
    {\rm subject\ to}:& \quad  \|\vw\|_
    0=2 \nonumber \\
                     &  \quad  w_0 =1,
\end{align}
where $I_{i,m}$ is the total interference when receiving signal $m$:
% \begin{align}
%   I_{i,m}=P_i\sum_{m'\neq m} \left| \vw^H \vh_{i,i,m'}\right|^2+
%   P_j\sum_{m} \left| \vw^H \vh_{i,j,m}\right|^2
% \end{align}
\begin{align}
   I_{i,m}=\sum_{j \neq i}
   P_j\sum_{m} \left| \vw^H \vh_{i,j,m}\right|^2.
\end{align}

Let
\begin{align}
    \mW=\left[\vw_1,\ldots,\vw_t \right].
\end{align}
By construction, $\mW \in {\mathbb C}^{N_r \times t}$ is a matrix with $1$'s in the first row and another single non-zero element in each column, thus implicitly defining the selected antennas.
Let
\begin{align}
    \mH_i=\left[\vh_{i,1},\ldots \vh_{i,t} \right]
\end{align}
be the equivalent ISI channel for all the transmitted signals.
Treating interference as noise, for every user $i$ the equivalent MIMO channel is now given by:
\begin{align}
    \vy_i=\mW^H \mH_i \vx_i + \mW^H\sum_{j \neq i}  \left(\mH_{i,j}\vx_j+ \vz_i\right).
\end{align}
The corresponding mutual information is now readily  computable. Furthermore, the equivalent direct channel $\mW^H \mH_i$ is diagonally dominant with high probability. Therefore, simple MIMO decoding techniques can be employed.

In the closed loop scenario, the transmitter knows the equivalent channel and may apply covariance shaping to maximize the mutual information for the equivalent channel \eqref{eq:IC_MIMO}. Furthermore, the transmission architecture will be greatly simplified by using the  singular value decomposition to the resulting effective channel (including the covariance shaping matrix at the transmitter and the noise whitening matrix at the receiver).

\section{Simulations}
\label{sec:sim}
The performance of the proposed scheme in several scenarios is studied in this section via simulation.

In the first set of simulations, we tested the robustness of the scheme for the basic $1/1/N_r/2$  LOS interference channel of Section~\ref{sec:LOS_int_channel}, where we also allow for some directional errors. For SNR values ranging from $-5$ dB to $20$ dB, we generated $100$ realizations of a LOS four-user interference channel. We repeated the experiment for three values  $d_{\max}=50,100,500\gl$ which is a reasonable number for practical massive MIMO scenarios. We evaluated the mutual information of user $1$, with all transmitters randomly located at directions chosen between $0$ and $180$ degrees. All interferers were assumed to be received with the same power. %The $gl/2$ array has no practical capacity with only nulling or non-linear interference suppression, so
As a benchmark for comparison we took non-naive time-division multiple access (TDMA), with two users transmitting per time slot, assuming a receiver that applies MMSE nulling of the undesired signal.
We calculated the average achievable rate over all the channel realizations, optimized over $d,\phi$ using a full search with $1^o$ resolution in $\phi$ and a $\gl/2$ uniform linear array. To test for robustness, we also evaluated the performance of a mismatched receiver suffering from i.i.d receiver directional errors with $\gs_{\gth}=0.1^o, 0.05^o, 0.01^o$, for $d_{\max}=50,100,500\gl$ respectively.  The results are depicted in Figures~\ref{fig:vsnr_d50}-\ref{fig:vsnr_d500}. The interference-free rate is nearly attained up to an SNR of roughly $10$ dB for $d_{\max}=100\gl$, and even up to roughly $20$ dB for $d_{\max}=500\gl$.
The slowing of the growth of the rates attained by the proposed  scheme is due to the limited size of the array. It is seen that the scheme demonstrates reasonable sensitivity to small errors in direction estimation.

Finally, we tested the performance of the scheme in a scenario of a $6$-user  interference channel, setting $d_{\max}=200\gl$. As expected, the gain over TDMA is smaller. This suggests that combining ergodic nulling with TDMA may be beneficial, particularly at high SNR.

To test the dependence of the achievable rates on $d_{\max}$, we chose a fixed SNR of $10$ dB and computed the achievable rate as a function of $d_{\max}$. The results are depicted in Figure~\ref{fig:vd}. While attaining the interference-free rate requires a separation of roughly  $100 \gl$, very significant performance gain over non-naive TDMA are achieved even at $d=30\gl$, for which a rate gain of  $50\%$ is achieved.
\begin{figure}[t]
\centering
\includegraphics[width=0.85\columnwidth]{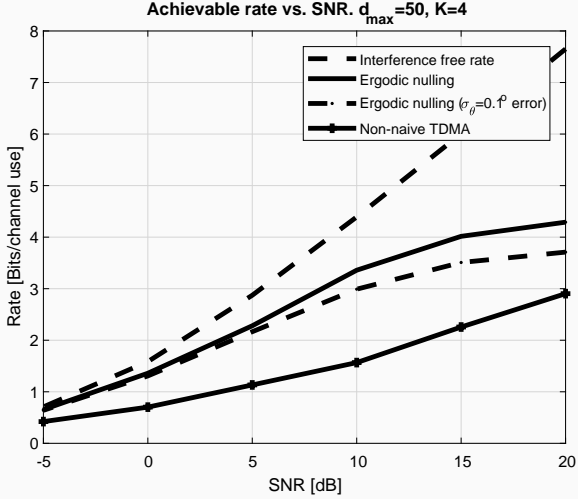}
\caption{Four-user interference channel with SIR=-$5$ dB and $d_{\max}=50 \gl$. The directions of the users is uniform and $100$ random channel realizations are drawn.}
\label{fig:vsnr_d50}
\end{figure}
\begin{figure}[t]
\centering
\includegraphics[width=0.85\columnwidth]{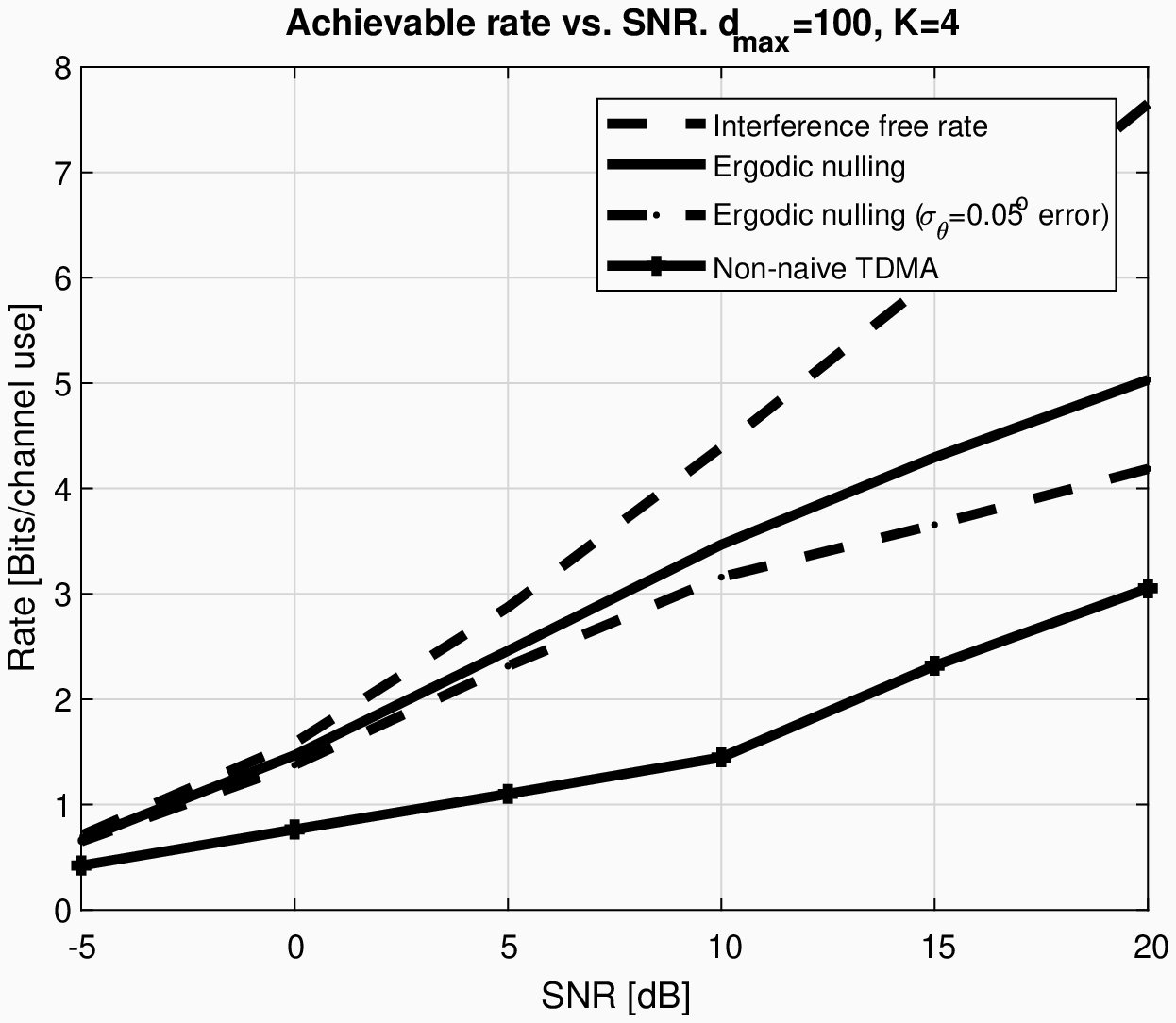}
\caption{Four-user interference channel with SIR=-$5$ dB and $d_{\max}=100 \gl$. The directions of the users is uniform and $100$ random channel realizations are drawn.}
\label{fig:vsnr_d100}
\end{figure}
\begin{figure}[t]
\centering
\includegraphics[width=0.85\columnwidth]{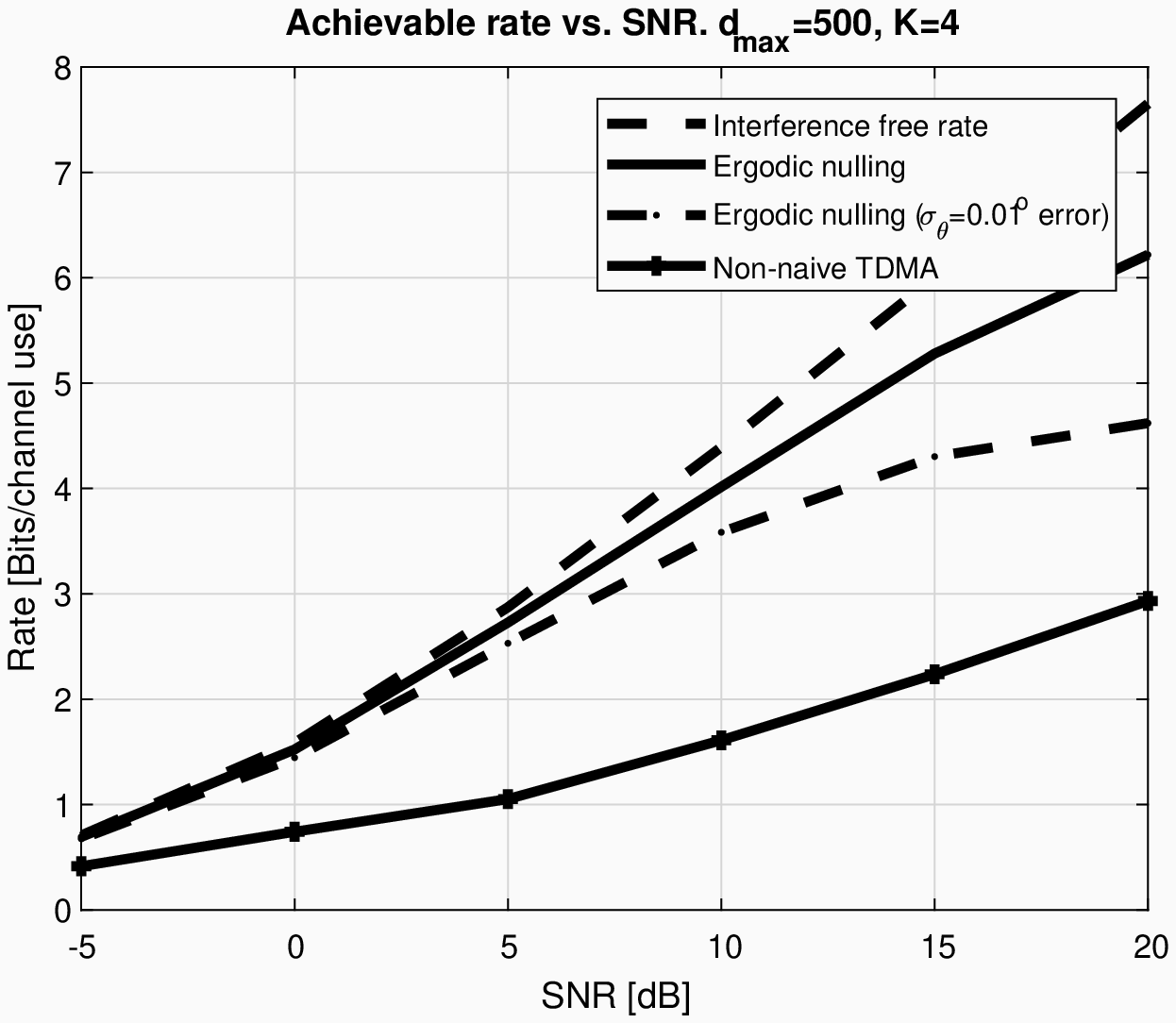}
\caption{Four-user interference channel with SIR=-$5$ dB and $d_{\max}=500 \gl$. The directions of the users is uniform and $100$ random channel realizations are drawn.}
\label{fig:vsnr_d500}
\end{figure}
\begin{figure}[t]
\centering
\includegraphics[width=0.85\columnwidth]{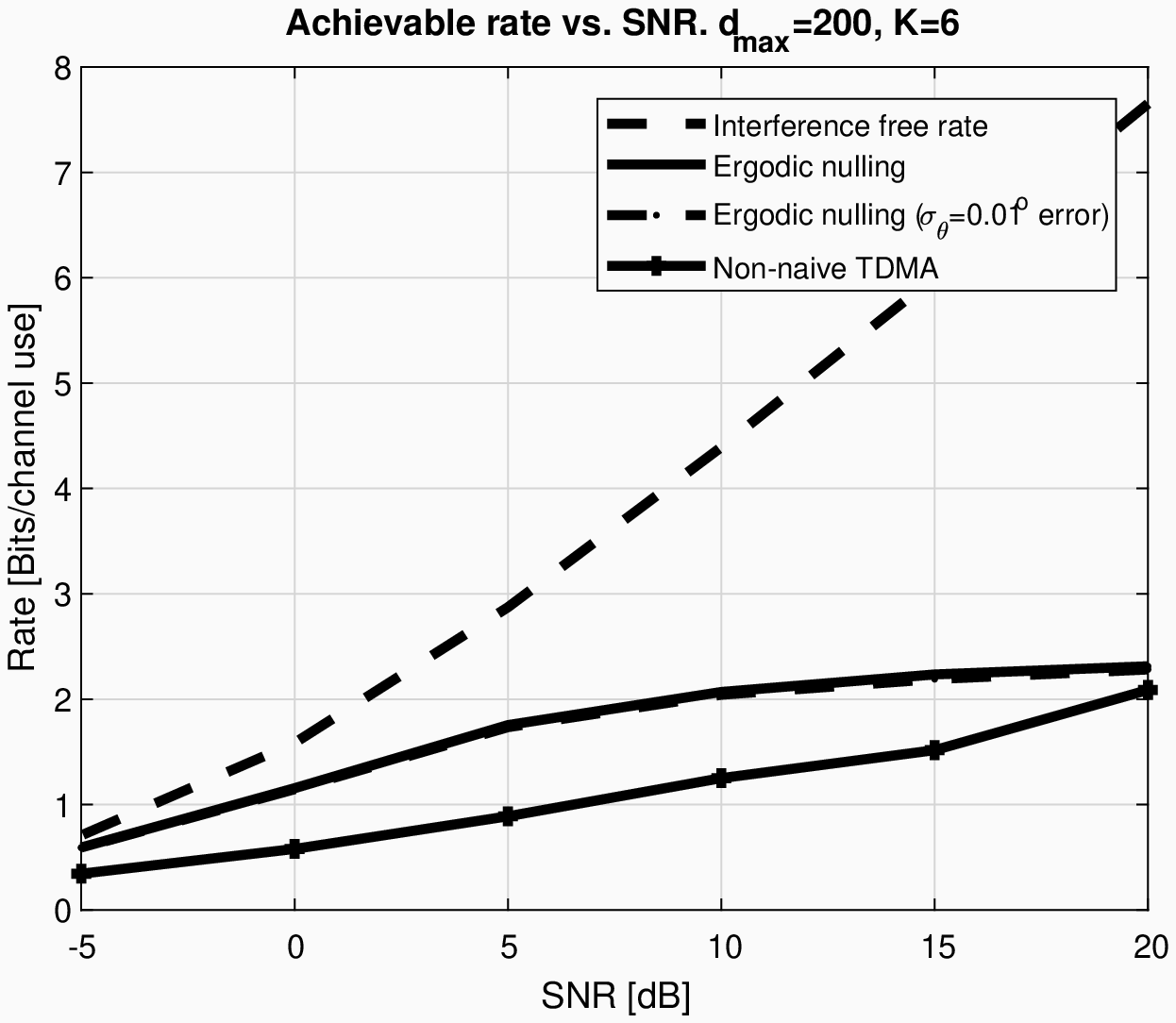}
\caption{Six-user interference channel
with SIR=-$5$ dB and $d_{\max}=200 \gl$.
The directions of the users is uniform and $100$ random channel realizations are drawn.}
\label{fig:vsnr_N6d200}
\end{figure}

\begin{figure}[t]
\includegraphics[width=0.9\columnwidth]{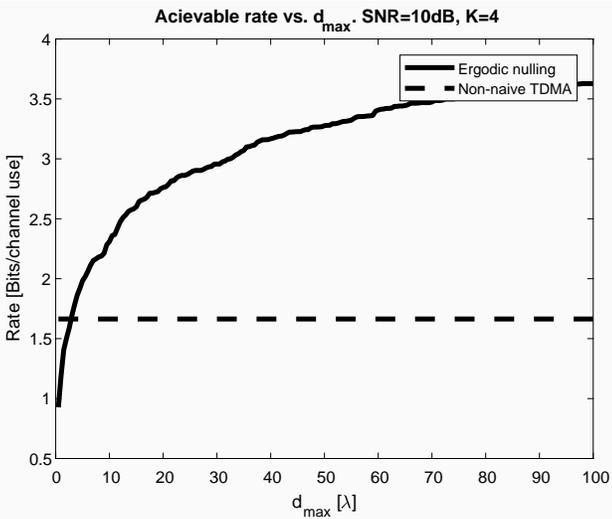}
\caption{Four-user interference channel where $100$ random channel realizations are drawn. 100 random channels. SNR=$10$ dB, SIR=-$5$ dB.}
\label{fig:vd}
\end{figure}

Finally, we simulated the performance attained in the scenario of a MIMO interference channel as considered in Section~\ref{sec:MIMO}. We evaluated the achievable rate of one user in a symmetric $2/2/N_r$ 3 -user-interference channel. We assumed that two reflections of the desired signals are received in addition to two other uncoordinated interfering signals. Note that with three receive antennas, spatial nulling will fail as the total number of streams (desired and undesired) is four. Since we do not assume CSI to be available at the transmitter, known interference alignment techniques are inapplicable. We considered two  benchmarks for comparison. The first is the interference-free rate corresponding to  the selected three antennas; the second is the maximal rate achieved via optimal selection of any three antennas from the array. We assumed that all paths are received with equal power, while the transmitter employs a $\gl/2$-spaced two-antenna array and isotropic transmission.
%transmitting omni-directionally.

The DoD at the transmitters of all desired and interfering signals as well as the DoA at the receiver were chosen uniformly at random. The simulations were carried out for $2$ and $4$ external interferers.
%, operating continuously.
We have used both a full-search algorithm and the proposed simplified two-step optimization procedure. Three antennas out of $N_r=250$ antennas were selected. Figure~\ref{fig:MIMO4} depicts the results for the case of a $2 \times 3$ MIMO system with 2 randomly located interferers. We present the achievable rates and the interference-free rates, for the optimal selection of $3$ antennas as well as for the simplified antenna selection scheme with optimal linear beamforming. It is clearly seen that for optimal selection, near interference-free rates are attained. The simplified selection technique is also near-optimal up to an SNR of about $20$ dB. We repeated the experiment with $4$ external interferers, a scenario where standard MIMO techniques
%based on algebraic structure only
are expected to yield very poor results.
Again, %We can clearly see that the
optimal selection yields almost interference-free rates as implied by the main theorem. Moreover, even the simplified selection technique achieves $10$ bits per channel use. This amounts to almost $70\%$ of the interference-free rate % with optimal selection at signal-to-noise-ratio of
at an SNR of $20$ dB.
\begin{figure}[htbp]
\includegraphics[width=0.9\columnwidth]{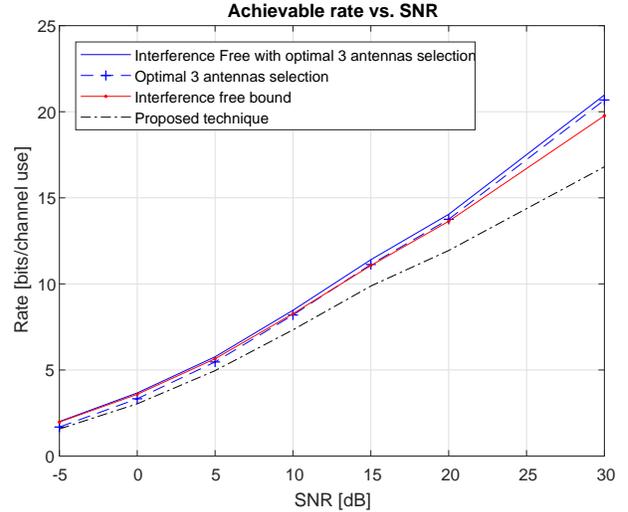}
\caption{$2 \times 3$ MIMO system with 2 external interferers. 50 random channels. $d_{\max}=250$.}
\label{fig:MIMO4}
\end{figure}

\begin{figure}[htbp]
\includegraphics[width=0.9\columnwidth]{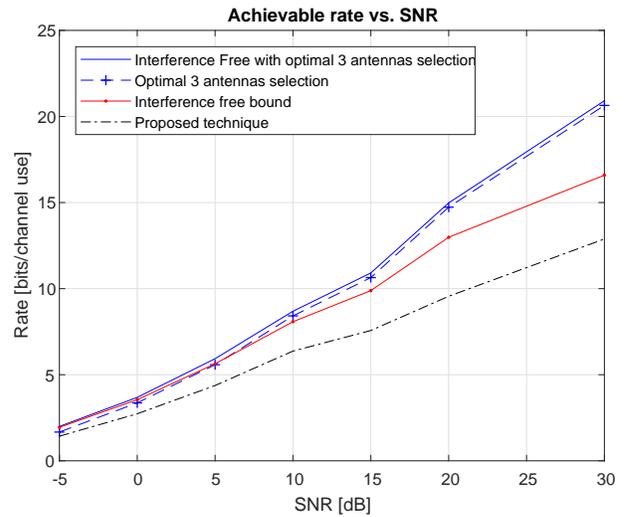}
\caption{$2 \times 3$ MIMO system with 4 external interferers. 50 random channels. $d_{\max}=250$.}
\label{fig:MIMO4}
\end{figure}

\section{Discussion}
\label{sec:disc}

We proposed a novel technique for interference suppression over   line-of-sight and specular multipath interference channels. The approach is based on  judiciously setting the distance between  two receive antennas to attain a beamforming vector with approximate nulls in the direction of the interferers. This can be implemented using antenna selection applied to a large linear array and the results were presented in this context.
It is important to note that adjusting antenna separation may equally be accomplished by other means. Examples include, e.g., rotating or moving array, or antenna selection applied to general arrays (e.g., circular or rectangular).

The main theorem shows that we can achieve half the degrees of freedom afforded by the system.
A significant advantage of the proposed approach compared to traditional interference alignment techniques is  that it only requires receive-side CSI.

%and the transmitter might include only a single transmit chain. Our
% Thus, the scheme does not require CSI information at the transmitter beyond transmission rate.  Since the scheme is applicable when there is a single transmit antenna per user, it allows operating in a regime where transmit ZF is impossible.

% In practice, moving the antennas to set the desired separation may be difficult to implement.  To overcome this, one possibility is to use the standard approach taken in massive MIMO systems, where two antennas of  a large array are switched into two receiver chains; see, e.g., \cite{gou2011aiming}.
%Since much of the complexity and power consumption involved in beamforming is caused by the low noise amplifiers, this solution is very relevant as a practical massive MIMO solution for the interference channel.

In a practical implementation, it is preferable to limit the dimensions of the receive antenna array. To that end, a receiver could  divide the interferers into two groups, a small group of strong interferers
for which approximate nulling is required  and a residual that is treated as noise. Moreover, from a system perspective, the users could be partitioned into   groups in which the number of strong interferers is limited.

Since the proposed approach is capable of suppressing any (finite) number of interferers, it  is  applicable also for non-symmetric interference channels with a configuration of  $t_i/N_{T,i}/N_{R,i}/r_i$, $i=1,...,K$, as long as all $N_{R,i}$ are large enough and for all $i$ we have $r_i \geq t_i+1$.

Similarly, the results can be easily extended  to configurations of the  interference channel with  $t \geq r+1$  as long as $N_t$ is sufficiently large and $N_r\geq r$, provided that the directional CSI is available at the transmitters.

Finally, we note that the advocated approach easily extends to the model of an interference multiple-access channel. Namely, given $r$ receive chains, $r-1$ single-antenna users can be afforded a full DoF while suppressing an arbitrary number of interferers, thus yielding a DoF utilization factor of $1-\nicefrac{1}{r}$.  A dual result holds for  the downlink.

\bibliographystyle{IEEEtran}
%\bibliography{mybib,symICbib}

\end{document}